\documentclass[journal]{IEEEtran}
\usepackage{amsmath}
\usepackage{amsfonts}
\usepackage{color}
\usepackage{bm}
\usepackage{graphicx}
\usepackage{multirow}
\usepackage{diagbox}
\usepackage{threeparttable}
\usepackage{algorithm,algorithmic}
\usepackage{multirow}
\usepackage{cite}
\usepackage[table,xcdraw]{xcolor}

\usepackage{amsmath}
\usepackage{amssymb}
\usepackage{amsthm}
\usepackage{subfigure}
\usepackage{textcomp}
\usepackage{xcolor}
\usepackage{booktabs}
\usepackage{url}
\newtheorem{definition}{Definition}

\newtheorem{lemma}{Lemma}

\newtheorem{theorem}{Theorem}

\usepackage{color}

\begin{document}
\title{Efficient Unit Commitment Constraint Screening\\ under Uncertainty}
\author{Xuan He$^1$, %~\IEEEmembership{Student Member,~IEEE,} 
 Honglin Wen$^2$, 
 Yufan Zhang$^3$, 
 Yize Chen$^4$, 
 and Danny H. K. Tsang$^1$%,~\IEEEmembership{Member, ~IEEE,} 
         \thanks{X. He and D.H.K. Tsang are with the Informuation Hub, Hong Kong University of Science and Technology (Guangzhou). H.Wen is with Shanghai Jiao Tong University. Y. Zhang is with University of California San Diego. Y. Chen is with University of Alberta, email: yize.chen@ualberta.ca.}\vspace{-15pt}}
        % <-this % stops a space

\markboth{In submission}%
{Shell \MakeLowercase{\textit{et al.}}: Bare Demo of IEEEtran.cls for IEEE Journals}

\maketitle

\vspace{-10pt}\begin{abstract}
Day-ahead unit commitment (UC) is a fundamental task for power system operators, where generator statuses and power dispatch are determined based on the forecasted nodal net demands. The uncertainty inherent in renewables and load forecasting requires the use of techniques in optimization under uncertainty to find more resilient and reliable UC solutions. However, the solution procedure of such specialized optimization may differ from the deterministic UC. The original constraint screening approach can be unreliable and inefficient for them. Thus, in this work we design a novel screening approach under the forecasting uncertainty. Our approach accommodates such uncertainties in both chance-constrained (CC) and robust forms (RO), and can greatly reduce the UC instance size by screening out non-binding constraints. To further improve the screening efficiency, we utilize the multi-parametric programming (MPP) theory to convert the underlying optimization problem of the screening model to a piecewise affine function. A multi-area screening approach is further developed to handle the computational intractability issues for large-scale problems. We verify the proposed method’s performance on a variety of UC setups and uncertainty situations. Experimental results show that our robust screening procedure can guarantee better feasibility, while the CC screening can produce more efficient reduced models. The average screening time for a single line flow constraint can be accelerated by 71.2X to 131.3X using our proposed method.
\end{abstract}


\section{Introduction}
\label{sec1}
%% Labels are used to cross-reference an item using \ref command.
Solving the unit commitment (UC) problem in an efficient manner is a fundamental yet challenging task for power system operators~\cite{9617122, bertsimas2012adaptive}. On the one hand, with the proliferation of renewable generation along with demand-side innovations such as electric vehicles (EVs) and smart thermostats, obtaining reliable UC solutions is a nontrivial task under an increasing level of forecasting uncertainties~\cite{conejo2010decision}. On the other, UC instances can be NP-hard due to their mixed integer programming (MIP) nature~\cite{pia2017mixed}, particularly in large-scale systems with numerous transmission constraints, which may render the UC model intractable.

Traditionally in the day-ahead stage, unit commitment using deterministic forecasting has been implemented and worked well for bulk power systems in terms of both solution quality and solution time~\cite{zhao2013unified}. Yet with increasing uncertainty, this approach often results in violations of system security constraints during real-time operation, or giving uneconomical decisions due to mismatch between forecasts and realizations. It is thus of operational and economic benefit to \emph{explicitly account for the uncertainty} during operational planning. Thus both industries and research have been looking into UC under stochastic load and generations~\cite{zheng2014stochastic}. Researchers have resorted to optimization techniques such as chance-constrained optimization~\cite{geng2019chance} and robust optimization~\cite{bertsimas2012adaptive, zeynali2023distributionally} to explicitly consider the uncertainty and most of them may be accompanied by specialized recourse policy \cite{roald2023power}. Their solution procedure can be different from the deterministic UC but still computationally challenging, while there are limited discussions on acceleration for such a solving process.

% can further complicate the already time-consuming deterministic UC problems, and hinder the implementation of such techniques in practice
In this paper, we address the following research question:

\emph{How can we accelerate UC solution process while explicitly considering forecasting uncertainty?}

To investigate this question, we resort to finding \emph{reduced problems} for the uncertain UC problems. In particular, our technique is based on the observation that in real-world systems, only a subset of security constraints such as line flow limits are binding or active, while eliminating the non-binding constraints would not change UC problem's solution. Identifying an appropriate subset of line flow limits is known as \emph{constraint screening}. We aim to explore the influence of net demand uncertainty on constraint screening for uncertain UC formulations. In addition, we propose an efficient approach that can not only respect the uncertainty formulation in UC instances, but also greatly accelerate the screening procedure.

Existing work on constraint screening focuses on the deterministic UC formulation. And researchers have developed approaches to conduct screening more efficient or more sufficient~\cite{zhai2010fast,porras2021cost,roald2019implied,10298823}. \cite{zhai2010fast} proposes to solve a screening model maximizing the (bidirectional) line flow for every network constraint for a specific net demand. If the maximal line flow does not reach the line limit, this limit will be identified as redundant. \cite{porras2021cost, he2022enabling} add  UC cost bounds to such a screening model to eliminate more constraints. Some recent approaches aim to improve screening efficiency by using machine learning models, such as directly classifying constraint sets~\cite{zhang2019data, ardakani2018prediction}, predicting total system costs~\cite{pineda2020data, he2022enabling}, and some warm starting strategies~\cite{xavier2021learning, Cordero2022Warm-starting}. 

\begin{figure*}[tb] \label{fig: intro}
% \vspace{-3.5em}
    % \hspace{0.35cm}
	\centering
	\includegraphics[width=1\linewidth]{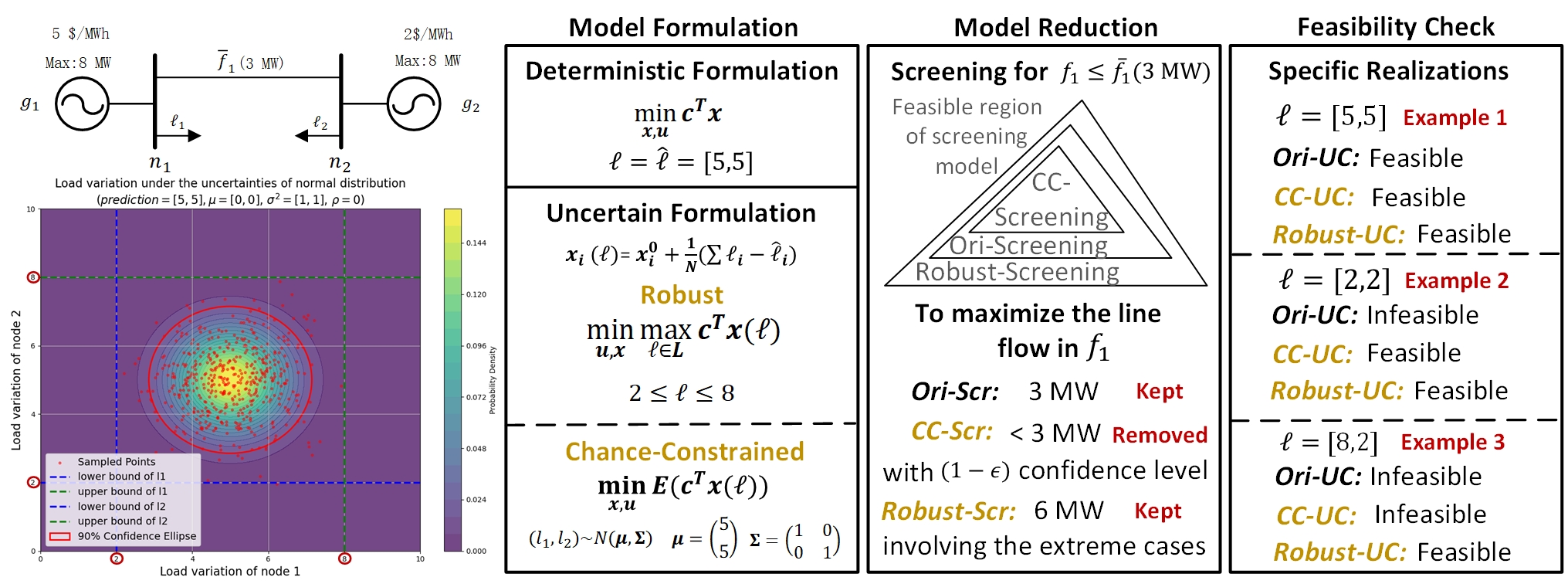}
    % \vspace{-0.5em}
	\caption{\footnotesize Our intuition based on a two-node UC problem under net demand uncertainty. The deterministic UC model, based on predicted net demand, may yield infeasible solutions (Example 2-3). The RO-UC model and CC-UC model are developed to obtain solutions that can be reliably adjusted to feasible ones using predefined recourse policies $\boldsymbol{x}(\boldsymbol{\ell})$, typically linear decision rules when the exact net demand is known. %The robust UC model can handle a broader region of uncertainties than the CC-model considering a specific distribution with an acceptable confidence level. 
 We propose to conduct constraint screening for different formulations to reduce their complexity. The robust screening can be valid for all formulations and the RO-UC can be feasible under small-probability events (Example 3), while the CC screening may achieve a smaller reduced model equivalent to the original CC-UC model (Example 1-2).}
\end{figure*}
% machine-learning-based
%are mostly derived from the deterministic UC problems and neglect the impact of the forecasting uncertainties. This is a crucial consideration for the UC problem, which is typically addressed using robust
 
However, previous screening methods cannot be applied to UC formulations under uncertainty such as robust (RO) \cite{zhao2013multi, moretti2020efficient} and chance-constrained (CC) models \cite{ozturk2004solution, wang2016risk}, as depicted in Fig. \ref{fig: intro}. Robust UC formulation is a computationally viable
methodology that provides solutions deterministically immune to
any realization of uncertainty within a defined set of uncertainties. CC-UC formulation can better utilize the prior knowledge of the uncertainty distribution, though it may fail in some extreme cases. These uncertain formulations are widely used in UC problems, but so far they have not been sufficiently investigated as to how the prevailing uncertainty would impact the screening results for these formulations. The existing consideration mainly develops different uncertainty sets in the screening model for the deterministic UC model \cite{roald2019implied, 10298823}, which requires resolving the reduced UC model for the exact realization. The resulting screening results may not be efficient or reliable for the uncertain UC formulations, especially for the CC-UC formulation with a specific distribution and a recourse policy that can directly adjust the initial generation schedule for an exact realization.

% and how to implement constraint screening accompanied by
% . The consequent concern is whether the inaccurate forecast of net demand can cause the constraint screening results to be invalid under different
By reformulating the screening model for the uncertain UC formulations considering the recourse policies, we can achieve a more suitable screening result. However, there still exists the computation burden arising from solving the optimization problem of the screening model for each line and each net demand instance. \cite{9314247} proposes to replace the optimization model solving with matrix operations, while in \cite{6938802}, screening time mainly depends on the number of non-redundant constraints rather than the number of lines. Data-driven approaches \cite{zhang2019data, ardakani2018prediction, he2023fast} can also speed up the screening procedure, but there may not be enough historical data for large-scale systems. Besides, these reduced UC models may not be equivalent to the original one \cite{zhang2019data, ardakani2018prediction, 9314247} or produce a more complicated screening model with considerable binary variables \cite{6938802}. This motivates us to find more efficient approaches to accelerate the screening procedure with a guarantee of equivalence.

In this work, we show that constraint screening is readily achievable for UC formulations with uncertain net demand involved. More importantly, we show both screening and UC problem-solving can be greatly accelerated using our proposed techniques. The main contributions of our work can be summarized as follows: 
\begin{itemize}
\item [1)] We are the first to conduct constraint screening for the UC formulations under uncertainty. RO-screening and CC-screening models are developed to reduce their original models reliably. The results show that the RO-screening can achieve 100\% feasibility for tested load samples, while the CC-screening can find more redundant constraints.

% We look into the formulation of robust and CC UC models and manage to develop the robust-screening and CC-screening models respectively. The reduced UC problem obtained via robust screening achieves a 0\% infeasibility rate, while for the CC screening problem, we achieve a lower number of active constraints along with infeasibility rates within 10\%.

\item [2)] A multi-parametric programming (MPP) method is adapted to convert each screening model to a piecewise affine function that maps the input of net demand to the maximal line flow. The results show that the screening procedure can be accelerated by 5.9X to 54.3X compared to directly solving the screening model.
% can be directly mapping and identify the redundant constraints.

% We convert each screening model to a piecewise affine function via the MPP method. Given the uncertainty representation of net demands, we can directly calculate the maximal flow for each line using the piecewise affine functions and identify the redundant constraints. The screening procedure can be accelerated by 5.9X to 54.3X.

\item [3)] A decomposition-based multi-area screening is proposed to accelerate the screening procedure reliably for large-scale systems. The results show that the total screening time can be reduced by 6.9\%, while the comparable scales and the equivalent solutions can be achieved for the reduced model compared to the cases without decomposition.

% The results show that decomposed screening can reduce the screening procedure time by approximately 2800 seconds. The reduced models retain, on average, 0.3\% more line limits while providing solutions that are 100\% equivalent to those obtained without using the decomposition approach.
\end{itemize}

%% Use \subsection commands to start a subsection.
\section{Preliminaries}
\label{sec:prelim}
% \textcolor{red}{I think it should be referred to as ``Preliminaries''}
\subsection{Deterministic UC model}
In the deterministic UC formulation, the system operators need to decide both the ON/OFF statuses (the commitment schedule) $\mathbf{u}$ as well as dispatch level  $\mathbf{x}$ for all generators to find the least total costs for time horizon $T$, corresponding to generators' cost vector $\mathbf{c}$. Without loss of generality, we denote $x_i(t)$ as the generator $i$
's generation at timestep $t$. We follow the typical modeling assumption to model power flows as a DC approximation, where $a_{i,j}$ denotes the entry in the Power Transfer Distribution Factors (PTDF) matrix~\cite{wang2012computational}. $\overline{\mathbf{f}}_j$ and $\underline{\mathbf{f}}_j$ denote the upper and lower bounds of line flow. We assume that each node has the generation and for the node without generator, we let the corresponding generation bound $\overline{x}_i =0$. $i, j$ denote the index of the bus and line, respectively. For a deterministic UC instance with $N$ generators involved, the problem can be formulated as 

% \todo{Typically we use capital letter to denote constant, small letter for variables}
\begin{subequations}
% \vspace{0.5em}
\label{UC}
    \begin{align}
\min _{\mathbf{u}, \mathbf{x}}\quad & \sum_{t=1}^{T}\sum_{i=1}^N  c_i x_{i}(t) \label{UC:obj_0}\\
\text { s.t. } \quad  &u_i \underline{x}_i \leq x_i(t) \leq u_i \bar{x}_i, \quad \forall i, t, \label{UC:gen_0}\\
 &\overline{\mathbf{f}}_{j} \leq \sum_{i=1}^{n}a_{i,j}(x_i(t)-\hat{\boldsymbol{\ell}}_i(t)) \leq \overline{\mathbf{f}}_{j}, \quad \forall j,t,\label{UC:flow_0}\\
 &\sum_{i=1}^{n}x_{i}(t) - \sum_{i=1}^{n}\hat{\boldsymbol{\ell}}_i(t) =0, \quad \forall t, \label{UC:balance_0}\\
 & u_i(t) \in \{0, 1\}, \quad \forall i,t, \label{UC:u_0}\\
& x \in \boldsymbol{X}_T. \label{UC: temproal}
\end{align}
\end{subequations}
where constraints \eqref{UC:gen_0}, \eqref{UC:flow_0} and \eqref{UC:balance_0} denote the generation bound, line flow limits and the system power balance, respectively.   \eqref{UC:u_0} enforces the binary constraint of the generator statuses, where $u_i=1$ indicates that the generator is on. \eqref{UC: temproal} denotes the temporal constraint set of generations such as ramping constraints. Note that in the deterministic case, $\hat{\boldsymbol{\ell}}_i(t)$ denotes the net demand considering the load and renewable generation, which is assumed forecasted perfectly and is regarded as the ground truth $\boldsymbol{\ell}_i(t)$, i.e., $\hat{\boldsymbol{\ell}}_i(t) = \boldsymbol{\ell}_i(t)$. 

% $N$ is the number of bus and $M$ is the number of lines. 
\subsection{Uncertainties from Net Demand and Affine Recourse Policy}
However, as mentioned above, operating conditions of the UC problem can be highly influenced by intermittent renewable generation and stochastic load, meaning $\hat{\boldsymbol{\ell}}_i(t) \neq \boldsymbol{\ell}_i(t)$. Thus, the solution obtained by \eqref{UC} can be infeasible for the ground truth $\boldsymbol{\ell}_i(t)$. To address this issue, UC problems considering variations in net demand of load and renewable generation are developed. For simplification, we only consider the single-step UC problem while the method can be flexibly extended to incorporate the multi-step UC problem. The relationship between the forecast, ground truth, and uncertainty is represented as
\begin{align} \label{UC:uncertainty}
\hat{\boldsymbol{\ell}} =\boldsymbol{\ell}+\boldsymbol{\omega}.
\end{align}

We assume that controllable generators can provide a response $\mathbf{x}(\boldsymbol{\omega})$ to adapt their generation
to the realization of the uncertainty $\boldsymbol{\omega}$. In particular, the response
from the generators can ensure that $\mathbf{x}(\boldsymbol{\omega})$ and the ground-truth $\boldsymbol{\ell}$ yield a solution that satisfies
(1d). We denote the total power mismatch due to prediction errors as $\Omega = \sum_{i=1}^{N}\omega_i$, which can be distributed among the generators based on participation factors $\alpha$ according to the following generation response policy:
\begin{align} \label{UC:ge_control}
x_i(\boldsymbol{\omega}) = x_i + \alpha_i\Omega.
\end{align}
where $\alpha_i$ is the participation factor to ensure that a given mismatch is balanced by the same amount of reserve activation, and we have $\sum_{i=1}^{N}\alpha_i = 1$.
% \begin{align} \label{UC:ge_sum_alpha}
% \sum_{i=1}^{n}\alpha_i = 1
% \end{align}
Typically, $\boldsymbol{\alpha}$ can be the decision variables to optimize or a specified vector \cite{8600344}. In this paper we consider the latter case, for the generator participating in balancing the mismatch, we have $\alpha_i = \frac{1}{|G|}$, where $|G|$ denotes the number of the participating controllable generators, and for the others $\alpha_i = 0$.
% By special design  $\mathbf{x}(\boldsymbol{\omega}), \mathbf{\boldsymbol{\ell}}$ yield a feasible solution, i.e., one that satisfies \eqref{UC:balance_0}. 

\subsection{Robust UC Model}
Given $\hat{\boldsymbol{\ell}}$, robust formulation includes all possible scenarios of uncertain renewable generation and demand forecasts, and optimizes the UC cost in the worst case, which results in the following min-max formulation:
%\todo{Normally use $\hat{\cdot}$ to denote forecast value. Let us follow the norm.}
\begin{subequations}\label{UC_robust}
% \vspace{0.5em}
\begin{align}
\min_{\mathbf{u},\mathbf{x}} \max_{\boldsymbol{\omega}} \quad & \sum_{i=1}^N c_ix_{i}(\boldsymbol{\omega}) \label{UC:ro_obj}\\
\text { s.t. } \quad
&u_i \underline{x}_i \leq x_i(\boldsymbol{\omega}) \leq u_i \bar{x}_i, \label{UC:ro_cons}\\
&\overline{\mathbf{f}}_{j} \leq \sum_{i=1}^{N}a_{i,j}(x_i(\boldsymbol{\omega})-\hat{\boldsymbol{\ell}}_i + \omega_i) \leq \overline{\mathbf{f}}_{j}, \label{UC:ro_flow}\\
&\sum_{i=1}^{N}x_{i}(\boldsymbol{\omega})-\sum_{i=1}^{N}(\hat{\boldsymbol{\ell}_i}-\omega_i)=0,\label{UC:ro_balance}\\
&u_i \in \{0, 1\} \label{UC:ro_u}.
\end{align}
\end{subequations}
In this work's robust formulation, we assume that the ground-truth load vector $\boldsymbol{\ell}$ satisfies the following box uncertainty set, while our method is actually generalizable to a variety of definitions of uncertainty set defined in the domain of robust optimization~\cite{bertsimas2012adaptive}: 
\begin{align} \label{UC:uncertainty_box}
\boldsymbol{\beta}_1 \cdot \hat{\boldsymbol{\ell}} \leq \hat{\boldsymbol{\ell}}-\boldsymbol{\omega}\leq \boldsymbol{\beta}_2\cdot\hat{\boldsymbol{\ell}}.
\end{align}
where $\cdot$ denotes the element-wise multiplication, and $\boldsymbol{\beta}_1$ and $\boldsymbol{\beta}_2$ are scalar vectors associated with each dimension of the forecasted load. Also, note that solving \eqref{UC_robust} is nontrivial and more time-consuming than solving the original UC problem \eqref{UC}. Thus it is of both research and practical interest to design an acceleration strategy. In the latter sections, we will describe how screening methods can reduce the constraint set of the robust UC formulation, and speed up the robust UC problem solving.

% \textcolor{red}{\eqref{UC:ge_control} ?},
% \begin{subequations}\label{UC_robust}
% % \vspace{0.5em}
% \begin{align}
% &~~~~\min_{\mathbf{u},\mathbf{x}} \quad \sum_{i=1}^n c_ix_{i} - \sum_{i=1}^{n}c_i\alpha_i\Omega \label{UC:obj}\\
% &~~~~u_i\underline{x}_i + \alpha_i\overline{p} \leq x_i \leq u_i\overline{x}_i - \alpha_i\overline{p}\\
% &~~~~-\overline{\mathbf{f}}_{j} \leq \sum_{i=1}^{n}a_{i,j}(x_i-\alpha_i\overline{p}-\hat{l}_i-p_i),\\
% &~~~~~~\overline{\mathbf{f}}_{j} \geq \sum_{i=1}^{n}a_{i,j}(x_i+\alpha_i\overline{p}-\hat{l}_i+p_i), \label{UC:ro_flow}\\
% &~~~~~~~~~~~~\sum_{i=1}^{n}x_{i}-\sum_{i=1}^{n}\hat{\boldsymbol{\ell}}_i = 0,\label{UC:ro_balance}\\
% &~~~~~~~~~~~~~u_i \in \{0, 1\} \label{UC:ro_u}.
% \end{align}
% \end{subequations}

\subsection{Chance-Constrained UC model}
Decisions induced by such robust problems can be conservative due to the considerations of worst-case scenarios, thus increasing both the generation reserves and system costs. As such extreme cases are the tails of forecasting distributions, in this paper, a general screening technique also applicable to the chance-constrained formulation of the UC problem is also investigated \cite{bienstock2014chance}. CC formulation explicitly limits the probability of constraint violations. Technically, chance constraints depict the maximum allowable violation probability of inequality constraints and reduce the feasible space of the UC problem to a desired confidence region. Mathematically, the CC-UC problem can be formulated as
\begin{subequations}
% \vspace{0.5em}
\label{UC_chance}
\begin{align}
\min _{\mathbf{u}, \mathbf{x}} \quad & \sum_{i=1}^N \mathbb{E}(c_i x_{i}(\boldsymbol{\omega})) \label{UC:obj}\\
\text { s.t. }\quad  &\mathbf{Pr}(x_i(\boldsymbol{\omega}) \geq u_i\underline{x}_i) \geq 1-\epsilon_x, \label{UC:chance_cons}\\
~~~~&\mathbf{Pr}(x_i(\boldsymbol{\omega}) \leq u_i \bar{x}_i) \geq 1-\epsilon_x, \\
~~~&\mathbf{Pr}(f_j(\boldsymbol{\omega}) \geq -\overline{\mathbf{f}}_{j}) \geq 1- \epsilon_f, \label{UC:chance_line}\\
~~~&\mathbf{Pr}(f_j(\boldsymbol{\omega}) \leq \overline{\mathbf{f}}_{j}) \geq 1- \epsilon_f, \label{UC:chance_line2} \\
~~~&f_j(\omega) = \sum_{i=1}^{N}a_{i,j}(x_i(\boldsymbol{\omega})-\hat{\ell}_i+\omega_i),\\
&\sum_{i=1}^{N}x_{i}(\boldsymbol{\omega})-\sum_{i=1}^{N}(\hat{\boldsymbol{\ell}}_i-\omega_i) = 0, \label{UC:chance_balance}\\
&u_i \in \{0, 1\} \label{UC:chance_u}.
\end{align}
\end{subequations}
The constraints on the controllable generations and line flows are enforced using the separate chance constraints \eqref{UC:chance_cons}-\eqref{UC:chance_line2}. In \eqref{UC:chance_balance}, the generation response $x_i(\boldsymbol{\omega})$ is selected in a way that maintains the power balance for the possible realization of uncertainty and the response of the system. The chance constraints guarantee that the constraint can satisfy a prescribed probability. The level of risk associated with the chance constraint can be regulated by selecting the probability of violation $\epsilon_x$ and $\epsilon_f$ \cite{8600344}. It is noteworthy that for the two formulations considered in this work, the RO-UC model can handle a larger region of uncertainties than the CC-model considering a specific distribution with an acceptable confidence level.

\section{Constraint Screening under Uncertainty}
As can be seen in Section \ref{sec:prelim}, all deterministic and uncertain UC formulations can be large-scale MILP problems that are cumbersome to solve within satisfactory timeframes. Meanwhile, there exist many redundant or inactive constraints as illustrated in Fig. \ref{fig: FR}, giving the potential to accelerate problem solving of \eqref{UC}, \eqref{UC_robust} and \eqref{UC_chance} by screening out such constraints. In this section, we show it is practical to screen out a large number of constraints in these uncertainty-aware UC formulations, which greatly improve solution efficiency for such UC problems.

To identify each redundant line limit under the deterministic case, previous optimization-based screening models \cite{zhai2010fast} evaluate whether line flows $f_{j}=\sum_{i=1}^{n}a_{i,j}(x_i-\hat{\boldsymbol{\ell}}_i)$ will hit the limits given the forecasted $\hat{\boldsymbol{\ell}}$: %In \eqref{screening2}, we describe such formulation, which can be used to conduct the original screening for the $j$-th line.
% \vspace{-0.7em}
\begin{subequations} \label{screening2}
\begin{align}
\max_{\mathbf{u}, \mathbf{x}}~ \text{or} \;\min _{\mathbf{u}, \mathbf{x}}\quad & f_j\\
\text { s.t. } \quad &\eqref{UC:gen_0}, \eqref{UC:balance_0},\\
 &\overline{\mathbf{f}}_{k} \leq \sum_{i=1}^{N}a_{i,k}(x_i-\hat{\boldsymbol{\ell}}_i) \leq \overline{\mathbf{f}}_{k}, \quad k \neq j, \label{Screening2:flow_0}\\
& 0\leq u_i \leq 1. \label{Screening2:u}
\end{align}
\end{subequations}
where $\hat{\boldsymbol{\ell}}$ is a known net demand vector for UC problem. Note that \eqref{Screening2:u} relaxes $u_i$ as continuous variables and thus \eqref{screening2} is a tractable linear programming problem. 

% \textcolor{red}{Shall this be $\boldsymbol{\hat{\ell}}$?}

Inspired by the derivation of the above screening model, in this work we develop the RO-screening and the CC-screening models and we are one of the first to look into such screening problems and to reformulate these models.

\begin{figure}[]
% \vspace{-3.5em}
    \hspace{0.35cm}
	\centering
	\includegraphics[width=0.99\linewidth]{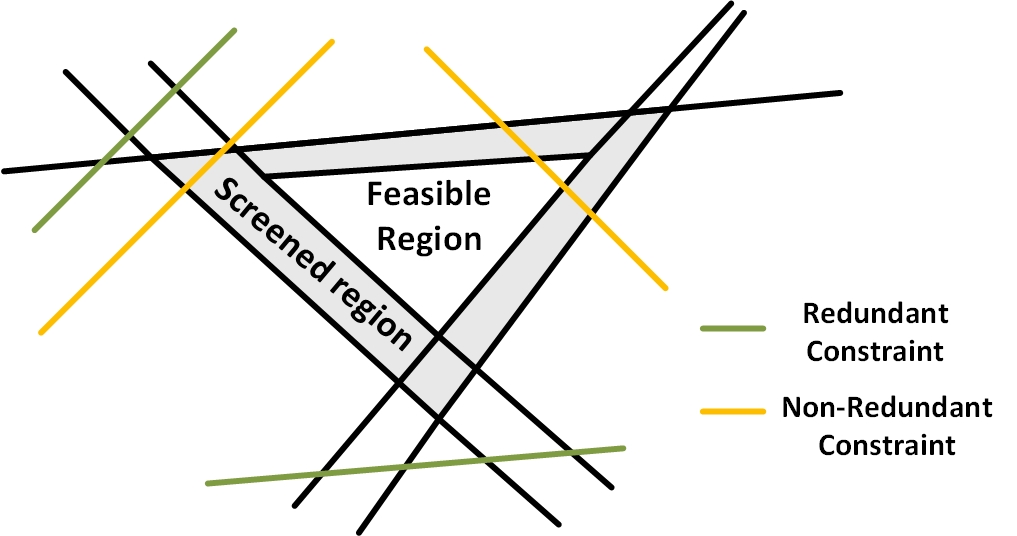}
	\caption{\footnotesize Redundant constraints identified by the screening models. Feasible regions are defined by all the UC model's constraints, while screened regions are defined by screening model constraints that are part of the UC model.}
	\label{fig: FR}
 
\end{figure}

\subsection{Robust UC Constraint Screening}
According to the optimization-based screening approach, the screening model for the robust UC model can be formulated as follows, 
\begin{subequations}\label{UC_Scr_ori_robust}
% \vspace{0.5em}
\begin{align}
\max_{\mathbf{u},\mathbf{x}, \boldsymbol{\omega}}~\text{or} ~\min_{\mathbf{u},\mathbf{x}, \boldsymbol{\omega}} \quad &f_j \label{UC:obj}\\
\text { s.t. } \;  & \eqref{UC:ro_cons}, \eqref{UC:ro_balance},\eqref{UC:uncertainty_box},\eqref{UC:ge_control},\\
 &-\overline{\mathbf{f}}_{k} \leq \sum_{i=1}^{N}a_{i,k}(x_i(\boldsymbol{\omega})-\hat{\boldsymbol{\ell}}_i + \omega_i) \leq \overline{\mathbf{f}}_{k}, \quad k \neq j, \label{UC:flow_ro}\\
&~0\leq u_i \leq 1. \label{UC: Screening_ro:u}
\end{align}
\end{subequations}
This formulation can identify whether the line flow will hit the limits for all possible realizations satisfying \eqref{UC:uncertainty_box} and \eqref{UC:ge_control}. To further simply this model and enable the screening results valid for more extreme cases where the recourse policy \eqref{UC:ge_control} is absent, the RO-screening model can be relaxed as
\begin{subequations}\label{UC_Scr_robust}
% \vspace{0.5em}
\begin{align}
\max_{\mathbf{u},\mathbf{x},\boldsymbol{\omega}}~\text{or} ~\min_{\mathbf{u},\mathbf{x},\boldsymbol{\omega}} \quad &f_j \label{UC:obj}\\
\text { s.t. } \;  & \eqref{UC:ro_cons}, \eqref{UC:ro_balance},\eqref{UC:uncertainty_box},\\
 &-\overline{\mathbf{f}}_{k} \leq \sum_{i=1}^{N}a_{i,k}(x_i-\hat{\boldsymbol{\ell}}_i + \omega_i) \leq \overline{\mathbf{f}}_{k}, \quad k \neq j, \label{UC:flow_ro}\\
&~0\leq u_i \leq 1. \label{UC: Screening_ro:u}
\end{align}
\end{subequations}

We also have the following Lemma to guarantee the feasibility conditions of the resulting screening model \eqref{UC_Scr_robust}:

\begin{lemma} \label{proof_robust}
Denote the non-redundant line limits of the robust model \eqref{UC_robust} as $S_{RO-REAL}$, and the non-redundant line limits identified by the screening model \eqref{UC_Scr_robust} as $\overline{S}_{RO}$. Then $S_{REAL-RO} \subseteq \overline{S}_{RO}$. That is, the constraint screening results obtained by \eqref{UC_Scr_robust} can guarantee the feasibility of \eqref{UC_robust} unchanged. 
% \todo{Too oral. Denote the problem first, avoid using making sure in a lemma.}
\end{lemma}
\begin{proof}
 Lemma \ref{proof_robust} can be proved by contradiction. Assume that $S_{REAL-RO} \nsubseteq \overline{S}_{RO}$, then this case can occur: for a particular $\hat{\boldsymbol{\ell}}-\Tilde{\boldsymbol{\omega}}$, $f_j=\sum_{i=1}^{N}a_{i,j}(\Tilde{x}_i$$(\Tilde{\boldsymbol{\omega}})-\hat{\ell}_i+\Tilde{\omega}_i)=\overline{f}_j$ can hold in \eqref{UC_robust}, while in \eqref{UC_Scr_robust}, the obtained maximum of $f_j=\sum_{i=1}^{N}a_{i,j}({x}_i(\Tilde{\boldsymbol{\omega}})-\hat{\ell}_i+\Tilde{\omega}_i) < \overline{f}_j$. 
 
 However, it can be seen that the feasible solution of \eqref{UC_robust} will always be feasible for \eqref{UC_Scr_robust}. This means in \eqref{UC_Scr_robust} the maximum of $f_j=\sum_{i=1}^{N}a_{i,j}({x}_i(\Tilde{\boldsymbol{\omega}})-\hat{\ell}_i+\Tilde{\omega}_i) \geq \sum_{i=1}^{N}a_{i,j}(\Tilde{x}_i(\Tilde{\boldsymbol{\omega}})-\hat{\ell}_i+\Tilde{\omega}_i)=\overline{f}_j$. Clearly, $S_{REAL-RO} \subseteq \overline{S}_{RO}$ can be proved by the contradiction. Thus, the screening results given by \eqref{UC_Scr_robust} can guarantee the feasibility of \eqref{UC_robust} unchanged.
\end{proof}
Similarly, when \eqref{UC:uncertainty_box} is satisfied, it can be proved that the non-redundant line limits of \eqref{UC} is a subset of $\overline{S}_{RO}$ so that the reduced UC model Testing 3 in Fig. \ref{fig: UC_Screening_Models} have the same feasibility situation.
\begin{figure*}[tb]
    \hspace{-1.5em}
	% \centering
	\includegraphics[width=1.0\linewidth]{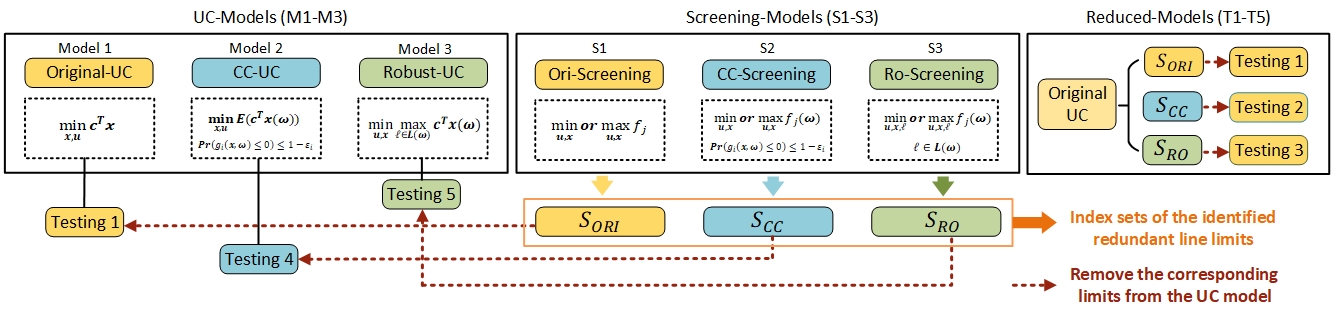}
	\caption{\footnotesize Relationship among the UC-Models, Screening-Models, and Reduced-Models under different uncertainty settings.} \label{fig: UC_Screening_Models}
\end{figure*}
\subsection{Chance-Constrained UC Constraint Screening}
Unlike the robust model, for the CC-UC model and its screening model, the uncertainty is formulated based on specific random distributions like the Gaussian distribution and the recourse policy is necessary to convert the chance constraints. We need to introduce a new set of optimization variables $r_i$, which represent the reserve capacity from each generation. Additionally, to guarantee adequate reserves for covering mismatches with high probability, we impose the following chance constraints:
\begin{subequations}\label{UC:r_chance}
\begin{align} 
&\mathbf{Pr}(\alpha_i\Omega \leq r_i) \geq 1 - \epsilon_x, \\
&\mathbf{Pr}(\alpha_i\Omega \geq -r_i) \geq 1 - \epsilon_x.
\end{align}   
\end{subequations}

We ensure that scheduled setpoints $x_i$ allow reserves $r_i$ without exceeding generation limits by enforcing:
\begin{equation} \label{chance_ri}
x_i - r_i \geq u_i \underline{x}_i,~ x_i + r_i \leq u_i \overline{x}_i. 
\end{equation}

Based on the above assumptions, the chance constraints on \eqref{UC:chance_line}-\eqref{UC:chance_line2} and \eqref{UC:r_chance} with a linear dependence on $\boldsymbol{\omega}$ can be exactly reformulated. Previous studies~\cite{bouffard2008stochastic, hodge2012comparison} indicated that the normal approximation of renewable generation and load forecasts were plausible in practice. In addition, the use of this assumption is because we are modeling a large number of geographically dispersed wind farms and electricity demands, and the law of large numbers holds~\cite{focken2002short}. In this work, we thus assume $\boldsymbol{\omega}$ follows a Gaussian distribution with mean $\mu_{\boldsymbol{\omega}}=0$ and known covariance matrix $\Sigma_{\boldsymbol{\omega}}$. The chance-constrained constraint screening model then can be given by: 
% \vspace{-1em}
\begin{subequations} \label{UC:chance_equai}
\begin{align}
\max _{\mathbf{u}, \mathbf{x}, \mathbf{r}}~\text{or} \min _{\mathbf{u}, \mathbf{x}, \mathbf{r}} \quad &  f_j(\boldsymbol{0}) \label{UC:obj}\\
\text { s.t. }\quad  &u_i\underline{x}_i + r_i \leq x_i \leq u_i\overline{x}_i - r_i, \\
&r_i \geq \alpha_i\Phi^{-1}(1-\epsilon_x)\sigma_\Omega, \label{UC:screen_CC_r}
\\
&\mathbb{E}(f_k(\boldsymbol{\omega})) \geq -\overline{\mathbf{f}}_k + \Phi^{-1}(1-\epsilon_f)\sigma_{f_k(\boldsymbol{\omega})}, \; k \neq j ,\label{UC:screen_CC_line}\\ 
&\mathbb{E}(f_k(\boldsymbol{\omega}))\leq \overline{\mathbf{f}}_k- \Phi^{-1}(1-\epsilon_f)\sigma_{f_k(\boldsymbol{\omega})}, \; k \ne j ,\label{UC:screen_CC_line2}
\\&\mathbb{E}(f_k(\omega)) = \sum_{i=1}^{N}a_{i,k}(x_i-\hat{\ell}_i),\\
&\sum_{i=1}^{N}x_{i}-\sum_{i=1}^{N}\hat{\ell}_i = 0, \label{UC:chance_balance_screening}\\
& 0 \leq u_i \leq 1. \label{UC:chance_u_screening}
\end{align}   
\end{subequations}
where $\mathbb{E}(f_j(\boldsymbol{\omega}))=f_j(\boldsymbol{0})$. $\sigma_\Omega$ and $\sigma_{f_j(\boldsymbol{\omega})}$ are given as $\sigma_\Omega^{2} = \mathbf{1}^T\Sigma_{\boldsymbol{\omega}}\mathbf{1},\; \sigma_{f_j(\boldsymbol{\omega})}^{2} = \sum_{i=1}^{N}a^{2}_{i,j}(\sigma_{\omega_i}^{2}+\alpha^{2}_{i}\sigma^{2}_{\Omega}) $. $\Phi^{-1}(\cdot)$ denotes the inverse Gaussian cumulative distribution, which is used to explicitly transform \eqref{UC:r_chance}, \eqref{UC:chance_line}, \eqref{UC:chance_line2} to \eqref{UC:screen_CC_r}-\eqref{UC:screen_CC_line2} respectively. The reliability of the chance-constrained screening for the chance-constrained UC model can be guaranteed according to the following analysis:
% Following the similar spirit of  Lemma \ref{proof_robust}, feasibility of chance-constrained screening can be also guaranteed.

\begin{lemma} \label{proof_cc}
Denote the non-redundant line limits of the chance-constrained model \eqref{UC_chance} as $S_{CC-REAL}$, and the non-redundant line limits identified by the screening model \eqref{UC:chance_equai} as $\overline{S}_{CC}$. Then $S_{CC-REAL} \subseteq \overline{S}_{CC}$. The constraint screening results obtained by \eqref{UC:chance_equai} can guarantee the feasibility of \eqref{UC_chance} unchanged. 
% \todo{Too oral. Denote the problem first, avoid using making sure in a lemma.}
\end{lemma}
\begin{proof}
 Lemma \ref{proof_cc} can be proved by contradiction. Assume that $S_{CC-REAL} \nsubseteq \overline{S}_{CC}$, then this case can occur: for a particular $\hat{\boldsymbol{\ell}}$, $f_j=\sum_{i=1}^{N}a_{i,j}(\Tilde{x}_i-\hat{\ell}_i)\geq\overline{f}_j$ can hold in \eqref{UC_chance}, while in \eqref{UC:chance_equai}, the obtained maximum of $f_j=\sum_{i=1}^{N}a_{i,j}({x}_i-\hat{\ell}_i) < \overline{f}_j$. 
 
 However, due to the fact that the feasible solution of \eqref{UC_robust} will always be feasible for \eqref{UC_Scr_robust}. This means in \eqref{UC:chance_equai} the maximum of $f_j= \sum_{i=1}^{N}a_{i,j}(\Tilde{x}_i-\hat{\ell}_i)\geq \overline{f}_j$. Then $S_{REAL-CC} \subseteq \overline{S}_{CC}$ can be proved by contradiction. Thus, the screening results given by \eqref{UC:chance_equai} can guarantee the feasibility of \eqref{UC_chance} unchanged.
\end{proof}

Solving the CC-UC screening model \eqref{UC:chance_equai} will lead to the reduced constraint set $\overline{S}_{CC}$ which will be non-redundant for the original CC-UC problem. Then it is safe to only include such a set for finding CC-UC's solutions.

% \textcolor{red}{Add a lemma here}

\section{MPP-Based Screening Acceleration}
% In this section, we show it is viable to solve the screening problem for UC formulation with uncertainties considered. Note that our approach differs from previous efforts on directly solving for each line's maximum and minimum flow to find binding lines. Solving such optimization-based screening model can be time-consuming, especially for networks with large number of line flow constraints. To address such challenge, we find multiparametric linear program (MPLP) can help reduce the optimization burden .
% \textcolor{red}{Give a bit more contexts here. Why solving this problem is important? And why we need acceleration?}

In the real-time operation stage, CC or robust UC models can be applied for the new-coming net demand instance to make the solution resilient to uncertainties. However, in the context of previous screening efforts, solving for each instance can be burdensome, especially for large-scale systems. To address this challenge, we propose the use of a multi-parametric linear programming (MPLP) approach, which is capable of handling varying parameters, such as the forecasted net demand. This approach accelerates the screening process by converting the screening models into affine functions.

% which \eqref{UC_Scr_robust} and \eqref{UC:chance_equai} can be converted to
\subsection{Mapping Operating Range to Optimal Line Flow}
MPLP is a method that enables the objective function and optimization parameters to be expressed as a function of parameters\cite{faisca2007multiparametric}. Indeed, the MPLP method can help find the mapping between the varying parameters of the UC model and the output variables. Our method draws insight from such a procedure, and given any parameter conditions in a predefined set, it is possible to identify the redundant constraints in UC models. Mathematically, consider general linear programming that the parameters appear on the right-hand side of the constraints,
\begin{subequations}\label{UC:MPP_LP}
\begin{align} 
\vspace{-5pt}\min_{\mathbf{x}} \quad &z: = \mathbf{a}^T\mathbf{y},\\
s.t.~~~ &\mathbf{A}\mathbf{y}\leq \mathbf{b} + \mathbf{F}\mathbf{\theta} \label{UC:MPP_LP_x}, \\
& \mathbf{\theta}  \in \mathbf{\Theta}.
\end{align}    
\end{subequations}
%\vspace{-3pt}
where $\mathbf{a}, \mathbf{A}, \mathbf{F}, \mathbf{b}$ are the coefficient vector or matrix with compatible dimensions. $\theta$ represents the varying parameters, such as the nominal load, renewable generation, electricity price, and others. Operators can typically deduce rough operating ranges for these parameters based on historical data. To utilize the MPLP, it is necessary to represent these patterns as a polyhedral convex set denoted by $\mathbf{\Theta}$.

% \textcolor{red}{Explain more on what is $\theta$.}

Assume that the optimal solution set of \eqref{UC:MPP_LP} is $\mathcal{Y}^{*}$, and that each optimal solution $\mathbf{y}^{*}(\theta) \in \mathcal{Y}^{*}$ is associated with the parameter $\theta$. By solving the MPLP problem, both critical region $\Theta \subseteq \mathbf{\Theta}$ and the parametric expression of $z^{*}(\theta)$ for $\Theta$ (See Fig. \ref{fig: MPP_region}) can be found. The definition of a critical region is as follows.
\begin{definition}
Let $\mathbf{C}$ denote the set of constraint indices in \eqref{UC:MPP_LP_x}, an optimal partition of $\mathbf{C}$ associated with parameter $\theta$ is the partion $(\mathbf{C}^{A}, \mathbf{C}^{I})$ where,
\begin{subequations}
\begin{align}
\mathbf{C}^{A}(\theta) \triangleq \{j\in \mathbf{C}| \mathbf{A}_{j}\mathbf{y}^{*}(\theta) = \mathbf{b}_{j} + \mathbf{F}_{j}\mathbf{\theta}, \forall \mathbf{y}^{*}(\theta) \in  \mathcal{Y}^{*}\},\\
\mathbf{C}^{I}(\theta) \triangleq \{j\in \mathbf{C}| \mathbf{A}_{j}\mathbf{y}^{*}(\theta)< \mathbf{b}_{j} + \mathbf{F}_{j}\mathbf{\theta},\exists \mathbf{y}^{*}(\theta) \in  \mathcal{Y}^{*}\};
\end{align}
\end{subequations}
where $\mathbf{A}_j, \mathbf{b}_j, \mathbf{F}_j$ are the jth rows of $\mathbf{A}, \mathbf{b}, \mathbf{F}$. For any $\mathbf{s} \subseteq \mathbf{C}$, let $\mathbf{A}_\mathbf{s}$ and $\mathbf{F}_\mathbf{s}$ be the submatrices of $\mathbf{A}$ and $\mathbf{F}$, consisting of rows indexed by $\mathbf{s}$. Then the critical region $\Theta_{\mathbf{s}_0}$related to $\mathbf{s}_{0} \subseteq \mathbf{C} $ can be defined as,
\begin{align}
\Theta_{\mathbf{s}_0} \triangleq \{ \theta \in \mathbf{\Theta}|\mathbf{C}^{A}(\theta)= \mathbf{s}_{0}\}.
\end{align}
which is the set of all parameters $\theta \in \mathbf{\Theta}$ with the same active constraints set $\mathbf{s}_0$ at the optimum(s) of problem \eqref{UC:MPP_LP} \cite{ji2016probabilistic}.
\end{definition}
Regarding the relationship between the feasible parameter space $\mathbf{\Theta}$ and critical regions, it can be presented as the following:
\begin{lemma} (\cite{borrelli2003geometric})
$\mathbf{\Theta}$ can be uniquely partitioned into critical regions $\{\Theta_{i}\}$ with the assumption that there is no (primal or dual) degeneracy.
\end{lemma}
Furthermore, the relationship between $z^{*}(\theta)$ and $\theta$ in \eqref{UC:MPP_LP} can be characterized as the following Theorem \ref{affine}.
\begin{theorem} \label{affine}
(\cite{borrelli2003geometric}) The objective function z*(·) is convex and piecewise affine over $\mathbf{\Theta}$ (in particular, affine in each critical region $\theta_{i}$). 
\end{theorem}

% The above theorems can be crucial to the development of the proposed constraint screening technique. In this paper, all calculations related to MPLP are performed using MPT3 toolbox \cite{6669862}.
\begin{figure}[t]
    \hspace{0.15cm}
	\centering
	\includegraphics[width=0.8\linewidth]{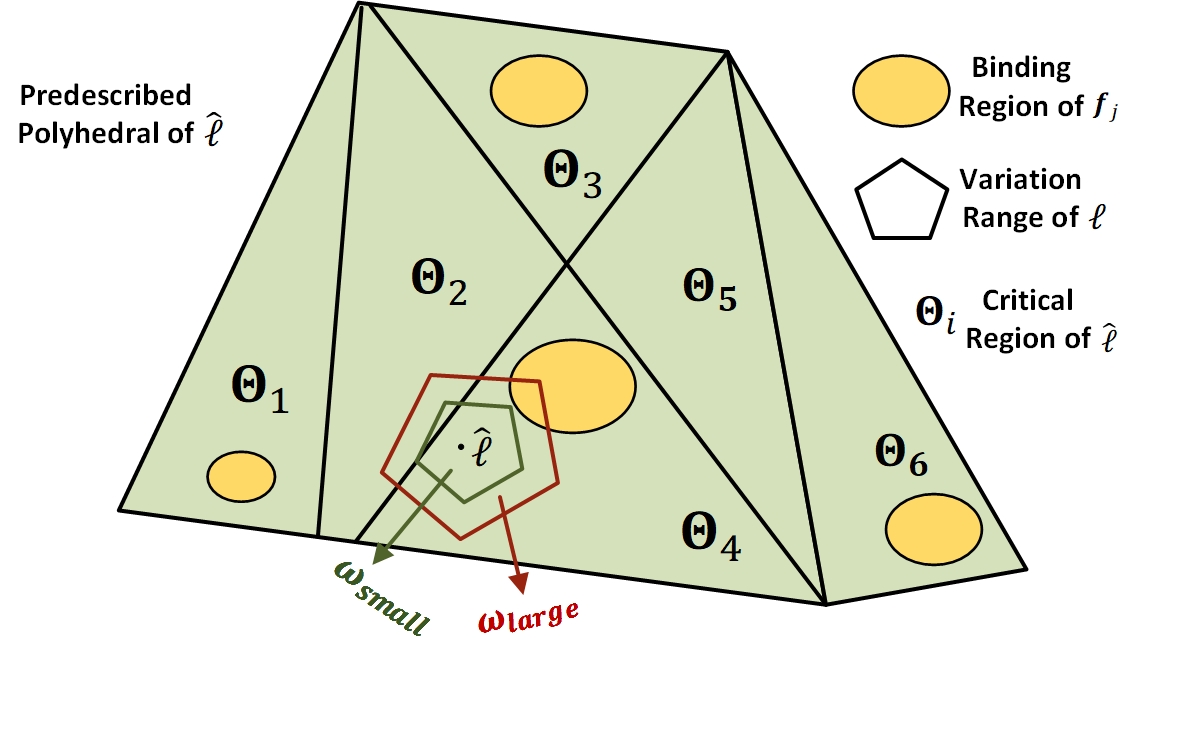}
 \vspace{-2em}
	\caption{\footnotesize Critical regions of forecasted net demand $\hat{\boldsymbol{\ell}}$ corresponding to one screening model for line $j$. In each critical region, there exists an affine function mapping $\boldsymbol{\hat{\ell}}$ to $f^{*}_j$. The binding region means that the net demand here can cause maximum line flow to reach the line limit. The uncertainty representation $\omega$ can help make the solution more reliable but at the cost of intersecting with more regions $\mathbf{\Theta}$.}
 %move the groundtruth $\boldsymbol{\ell}$ from non-binding to binding region and thus causing the screening result of this line limit invalid. }
	\label{fig: MPP_region}
\end{figure}

To develop the MPLP of the screening models, we consider the forecasted net demand $\hat{\boldsymbol{\ell}}$ as the varying parameter $\theta$, rather than a fixed value in the original screening models \eqref{UC_Scr_robust} and \eqref{UC:chance_equai}. Using historical data, the operator may have the ability to deduce the variation range of $\hat{\boldsymbol{\ell}}$ for the managed UC instances and represent this range as a polyhedral set $\hat{\boldsymbol{\Theta}}$, as illustrated in Fig. \ref{fig: MPP_region}. The corresponding MPLP for this setting can be formulated as follows:
% In the Robust-Screening \eqref{UC_Scr_robust} and CC-Screening \eqref{UC:chance_equai},
% \textcolor{red}{This is not easy to understand. Where does that prior knowledge come from? Why they want to use such load variation? How is the load variation range related to Robust screening and CC screening?}
% We treat the forecasted nominal load $\hat{\boldsymbol{\ell}}$ in pre-described polyhedral $\hat{\boldsymbol{\Theta}}$ as parameters as illustrated in Fig. \ref{fig: MPP_region}, and formulate the \emph{sample-agnostic screening} model as follows,
\begin{subequations}
\begin{align}  \label{UC:MPP_LP}
f^{*}_j(\mathbf{y}) =: \max _{\mathbf{y}}&~\text{or} ~\min _{\mathbf{y}} \quad f_j \\\
\text{s.t.}~\mathbf{y}&\in \mathcal{Y}{(\hat{\boldsymbol{\ell}})}, \\\hat{\boldsymbol{\ell}} &\in \hat{\boldsymbol{\Theta}}.
\end{align}    
\end{subequations}
Then, we can deal with the above model via an MPT3 toolbox \cite{6669862} offline to get the mapping $f^{*}_j(\hat{\boldsymbol{\ell}})$ between the maximum or minimum of $f_j$ and the incoming forecast of the nominal demand $\hat{\boldsymbol{\ell}}$. In each critical region $\hat{\boldsymbol{\Theta}}_i$, the mapping is a linear function,
\begin{subequations}\label{UC:MPP_LP}
\begin{align} 
f^{*}_{i,j}(\hat{\boldsymbol{\ell}}) = \hat{\boldsymbol{a}}_i^{T}\hat{\boldsymbol{\ell}} + \hat{\boldsymbol{b}}_i.
\end{align}    
\end{subequations}
where the derivation and computation of critical regions and piecewise affine functions can be found in \cite{faisca2007multiparametric}. Then, the online screening given the specific value of $\hat{\boldsymbol{\ell}}$ can be greatly accelerated by using $f^{*}_{i,j}(\hat{\boldsymbol{\ell}})$.

\subsection{Decomposition for the Large-Scale System}
% Though the MPLP techniques can avoid the online solution procedure of each screening model, it is still intractable for large-scale system with numerous line limits, which require considerable time and memory to get and store all the affine policies of screening models. Thus, we resort to decomposition approach  to reduce the complexity of the MPP derivation for large-scale system. 
% The scheduled injections can be treated as parameters $s^{g,h}_{b}$ assigned to proxy buses, representing the $b$-th pair of injection from region $g \in \mathbf{P}$ to region $h \in \mathbf{P}$.

In practice, large interconnected power systems are typically managed by independent system operators or regional transmission operators. Each operator has its operating area within which internal resources are used economically \cite{guo2017coordinated}. Then, the entire system can be decomposed into $N_{P}$ regions denoted by $\mathbf{P}$. The original PTDF-based formulation of line flow relies on all nodes in the entire system, which is hard to be decomposed into subsystems. Therefore, we reformulate the single-step UC into a decomposable model as follows,
%+ \sum_{h\in \mathbf{P}}\sum_{b\in \mathbf{B}_{p,h}}\pi^{p,h}_{b}s^{p,h}_{b} 
\begin{subequations} \label{UC_de}
% \vspace{0.5em}
\begin{align}
\min _{\mathbf{u}, \mathbf{x}, \boldsymbol{\delta}}\quad  &\sum_{g=1}^{N_P}\sum_{i=1}^{N_{g}}  c_i^gx_i^g\\
\text { s.t. } \quad  &u_i \underline{x}^g_i \leq x^g_i \leq u^g_i \bar{x}^g_i, \quad i= 1,2,...,N_g, \forall g \in \mathbf{P},\\
&-\overline{\mathbf{f}}_{j,k}^g \leq L^{g}_{j,k}=\frac{(\delta^g_j -\delta^g_k)}{R_{j,k}} \leq \overline{\mathbf{f}}_{j,k}^g, ~ (j,k) \in \mathbf{F}_{g}, \forall g \in \mathbf{P}, \label{UC:flow_g}\\
&\sum_{k \in F^{g}_{i,k}}L^{g}_{i,k} + x_i^g = \hat{\boldsymbol{\ell}}_i^g, \quad i= 1,2,...,N_g, \forall g \in \mathbf{P}, \label{UC:balance_q}\\
& u_i^g \in \{0, 1\}, \quad i= 1,2,...,N_g, \forall g \in \mathbf{P}, \label{UC: temproal_q}\\
&\delta^{ref} = 0. \label{UC: ref}
\end{align}
\end{subequations}
where $N_g$ is the number of buses in the area $g$. $R_{j,k}$ is the reactance of the line connecting buses $j$ and $k$. $g$ is the area index. $\mathbf{F}_g$ is the index set of lines in area $g$. $\delta^g_j$ and $\delta^g_k$ are the voltage angles of bus $j$ and bus $k$, respectively. Note that for a bus $i$ without any unit, $x^g_i$ will be 0. $\delta^{ref}$ is the reference bus.

% and if a bus $i \notin \mathbf{B}_{g,h}$, $s^{g,h}_{i}$ will be 0.
\begin{figure*}[t]
    % \hspace{0.15cm}
	\centering
	\includegraphics[width=0.9\linewidth]{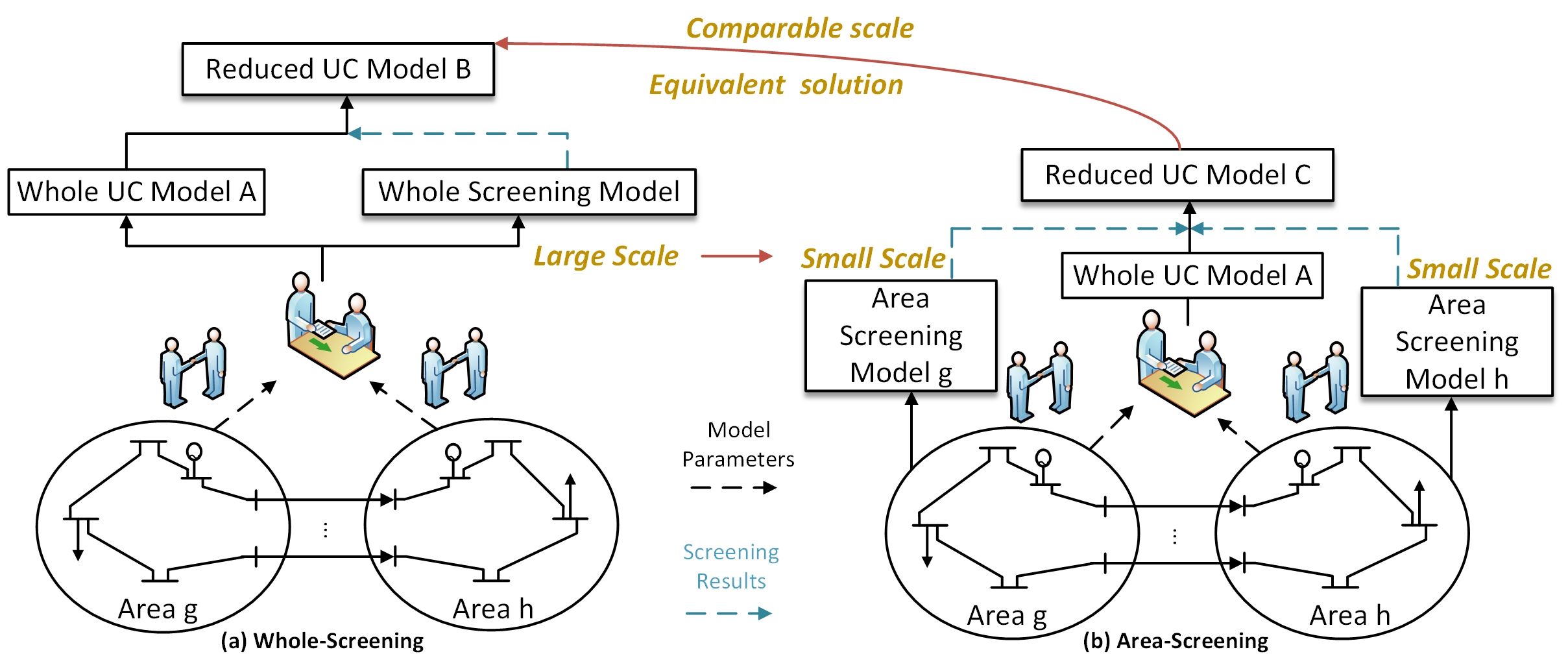}
	\caption{\footnotesize Schematic of Whole-Screening and Area-Screening approaches. Compared to the whole-screening model, the area-screening models only involves the local variables and constraints, which can be small scale and solved in parallel. Lemma \ref{lemma_area} guarantees that the resulted model C is equivalent to model B.}
 %move the groundtruth $\boldsymbol{\ell}$ from non-binding to binding region and thus causing the screening result of this line limit invalid. }
	% \label{fig: MPP_region}
\end{figure*}
For the whole system, consider the line in the area $g$ connected to the bus $w$ and bus $v$, the screening model will be
\begin{subequations} \label{UC_global_scr}
\begin{align} 
\min _{\mathbf{u}, \mathbf{x}, \boldsymbol{\delta}}~ \text{or}~ \max _{\mathbf{u}, \mathbf{x}, \boldsymbol{\delta}} \quad  &L^{g}_{w, v} \\
\text{s.t.} \quad &\eqref{UC:gen_0}, \eqref{UC:balance_q} , \eqref{UC: ref},\\
&-\overline{\mathbf{f}}_{j,k}^g \leq L^{g}_{j,k} \leq \overline{\mathbf{f}}_{j,k}^g, ~ (j,k) \in \mathbf{F}_{g/(w,v)}, \;\forall g \in \mathbf{P}, \label{UC:flow_scr_g} \\
&0 \leq u_i^g \leq 1, \quad i= 1,2,...,N_g, \; \forall g \in \mathbf{P}.\label{UC: temproal_q}
\end{align}
\end{subequations}
Further, if we only consider the line $(w,v) \in \mathbf{F}_g$ within area $g$, the screening model will be,
\begin{subequations} \label{Scr: area}
\begin{align} 
\min _{\mathbf{u}, \mathbf{x}, \boldsymbol{\delta}} \text{or} \max _{\mathbf{u}, \mathbf{x}, \boldsymbol{\delta}} \quad & L^{g}_{w, v}\\
&\text { s.t. } \quad  u_i \underline{x}^g_i \leq x^g_i \leq u^g_i \bar{x}^g_i, \quad i= 1,2,...,N_g, \label{UC:gen_scr_g}\\
&-\overline{\mathbf{f}}_{j,k}^g \leq L^{g}_{j,k} \leq \overline{\mathbf{f}}_{j,k}^g, ~ (j,k) \in \mathbf{F}_{g/(w,v)}, \label{UC:flow_scr_g}\\
&\sum_{k \in F^{g}_{i,k}}L^{g}_{i,k} + x_i^g = \hat{\boldsymbol{\ell}}_i^g, \quad i= 1,2,...,N_g, \label{UC:balance_scr_q}\\
&0 \leq u_i^g \leq 1, \quad i= 1,2,...,N_g.
\end{align}
\end{subequations}
Note that if the reference bus belongs to area $g$, \eqref{UC: ref} will be involved by the screening model \eqref{Scr: area}.
\begin{lemma} \label{lemma_area}
Denote the non-redundant line limits of the original UC model \eqref{UC_de} as $S_{ori}$, and the non-redundant line limits identified by the whole-screening model \eqref{UC_global_scr} as $\overline{S}_{whole}$ and the area-screening model \eqref{Scr: area} for the area $g$ as $\overline{S}_{g}$. Then $S_{ori} \subseteq \overline{S}_{whole} \subseteq \bigcup_{g=1}^{N_P}\overline{S}_{g}$. The constraint screening results obtained by \eqref{Scr: area} can guarantee the feasibility of \eqref{UC_de} unchanged. 
\end{lemma}

\begin{proof}
$S_{ori} \subseteq \overline{S}_{whole}$ has been proved in \cite{zhai2010fast} and thus we only need to prove that $\overline{S}_{whole} \subseteq \bigcup_{g=1}^{N_P}\overline{S}_{g}$. This can be also proved by contradiction. Assume that $\overline{S}_{whole} \nsubseteq \bigcup_{g=1}^{N_P}\overline{S}_{g}$, then the following case can occur: for a particular limit of $L^{g}_{w, v}$, it is identified as a redundant constraint by \eqref{Scr: area} while it is non-redundant for \eqref{UC_global_scr}. This means the maximum (or minimum) of \eqref{UC_global_scr} is larger (or lower) than that of \eqref{Scr: area}, which indicates that the optimal solution of \eqref{UC_global_scr} is not feasible for \eqref{Scr: area}. However, it can be seen that the constraint set of \eqref{Scr: area} is a subset of the constraint set of \eqref{UC_global_scr}. This means the feasible solution of \eqref{UC_global_scr} will be always feasible for \eqref{Scr: area}. Clearly, $\overline{S}_{whole} \subseteq \bigcup_{g=1}^{N_P}\overline{S}_{g}$ can be proved by the contradiction and thus $S_{ori} \subseteq \overline{S}_{whole} \subseteq \bigcup_{g=1}^{N_P}\overline{S}_{g} $ holds. Consequently, the screening results given by \eqref{Scr: area} can guarantee the feasibility of \eqref{UC_de} unchanged.
\end{proof}

Then, \eqref{Scr: area} can be solved as an MPLP to get the mapping from the tie-lie scheduling to optimal line flow. For the area $g$, the affine policy set $\mathcal{F}_g$ can be given as,
\begin{align} \label{UC:MPP_de_LP}
\mathcal{F}_g =: \{{L^{g,*}_{w, v}(\hat{\boldsymbol{l}}_g) = \hat{\boldsymbol{a}}_i^{T}\hat{\boldsymbol{l}}_g + \hat{\boldsymbol{b}}_i, i \in \boldsymbol{I}^{\Theta}_{g,j}, j \in \mathbf{F}_{g}\}};
\end{align}    
where $\boldsymbol{I}^{\Theta}_{g,j}$ is the index set of critical region of the line $j$ in area $g$. For each area $p \in \mathbf{P}$, we can get the affine policy set $\mathcal{F}_g$ and thus achieve the constraint screening for the whole large interconnected power system.  

The above multi-area formulations \eqref{UC_de} and \eqref{Scr: area} can be easily extended to the uncertain formulations as \eqref{UC_robust} and \eqref{UC_chance}, which are still the LPs that can be transferred to the affine policy for acceleration.

\section{Case studies}
In this section, we evaluate the performance of the proposed constraint screening method for UC problem under uncertainty and explore the impacts brought by net demand uncertainly upon the screening results over a wide range of problem settings.  We demonstrate the clear advantage of the proposed MPP-based screening procedure in both performance and efficiency.

\subsection{Simulation Setup}
We carry out the numerical simulations on 39-, 118- and 300-bus power systems to verify
the proposed method’s effectiveness and scalability. The configurations of the investigated systems are referred to as the corresponding cases in MATPOWER. We select 10 nodes to represent the critical nodes whose load variations are relatively larger and have more impact on the line flows. We assume that there exists the perturbation $\Tilde{\boldsymbol{\omega}}$ in the forecasted net demands $\hat{\boldsymbol{\ell}}$ of the selected nodes, and $\hat{\boldsymbol{\ell}}$ is in the predefined polyhedral. To realize different patterns of $\Tilde{\boldsymbol{\omega}}$, in the CC-UC cases, we assume $\epsilon_f=\epsilon_x=\epsilon$ and $\sigma_{\omega_i}$ with results given in Table \ref{table2} and Fig. \ref{num_s1_s3}-\ref{UC_binding_10_10}. In the robust screening cases, we test the settings of $\beta_1$ and $\beta_2$ given in Table \ref{table3}. 
% To evaluate the feasibility of reduced models derived by the CC screening model and robust screening model, we test 200 uncertainty realizations $\Tilde{\boldsymbol{\omega}}$ under each setting to calculate the infeasibility rate of Testing 2-4 illustrated in Fig. \ref{fig: UC_Screening_Models}.

% ($0.8\hat{\boldsymbol{\ell}}_0 \leq \hat{\boldsymbol{\ell}} \leq 1.2\hat{\boldsymbol{\ell}}_0$)

All simulations have been carried out on an unloaded MacBook Air with Apple M1 and 8G RAM. Specifically, all the optimization problems are modeled and solved using YALMIP toolbox and MPT3 toolbox in MATLAB R2022b.

For the sake of convenience, we use the expression given in Fig. \ref{fig: UC_Screening_Models} to refer to the involved models in the remaining part:
\begin{enumerate}
\item \textbf{UC models:} \texttt{M1} (Original), \texttt{M2} (CC) and \texttt{M3} (Robust).
\item \textbf{Screening models:} \texttt{S1} (Original), \texttt{S2} (CC) and \texttt{S3} (Robust).
\item \textbf{Index sets of redundant constraints:} $S_{ORI}$ (given by \texttt{S1}), $S_{CC}$ (given by \texttt{S2}) and $S_{RO}$ (given by \texttt{S3}).
\item \textbf{Reduced UC models:} \texttt{T1} (\texttt{M1} removes $S_{ORI}$), \texttt{T2} (\texttt{M1} removes $S_{CC}$), \texttt{T3} (\texttt{M1} removes $S_{RO}$), \texttt{T4} (\texttt{M2} removes $S_{CC}$),  \texttt{T5} (\texttt{M3} removes $S_{RO}$). \
\end{enumerate}
% \textcolor{red}{Make the axis' texts larger. Same for Fig. 7-8}} 

\subsection{Number of Screened Constraints}
We solve the screening models \texttt{S1-S3} for each line limit of 39-bus system and 118-bus system, and record the number of redundant constraints as shown in Fig. \ref{num_s1_s3}. It can be seen that \texttt{S2} tends to find the most redundant constraints, the numbers are 87 and 348. \texttt{S3} screens out the least, the numbers are 79 and 341. To see this, for the same forecast $\hat{\boldsymbol{\ell}}$, the screened region of \texttt{S2} is the smallest one and the maximum line flow over this region may not reach the removed line limit in \texttt{M2}. Meanwhile, the corresponding line flow can reach the original line limit in the region outside the screened region of \texttt{S2}. Then, a line limit can be redundant for \texttt{S2}, while it is non-redundant for \texttt{S1} and \texttt{S3}. On the contrary, the screened region of \texttt{S3} covers that of \texttt{S1}, so a line limit can be non-redundant for \texttt{S3}, while it is redundant for \texttt{S1}. This corresponds to the fact that robust optimization often considers a more conservative uncertainty set, and in our case, we have more constraints reserved. In Fig. \ref{UC_binding_10_10}, it can be found that even though the line limits and feasible region can be tightened in \texttt{S2}, the locations of the redundant constraints in the topology are still similar to \texttt{S1} and \texttt{S3}. Besides, most of the lines whose limits are non-redundant are connected or closer to the generators, e.g., line 33-19, line 35-22, and line 38-29, which are expected.
\begin{figure}[h]
    % \hspace{0.15cm}
	\centering
	\includegraphics[width=0.95\linewidth]{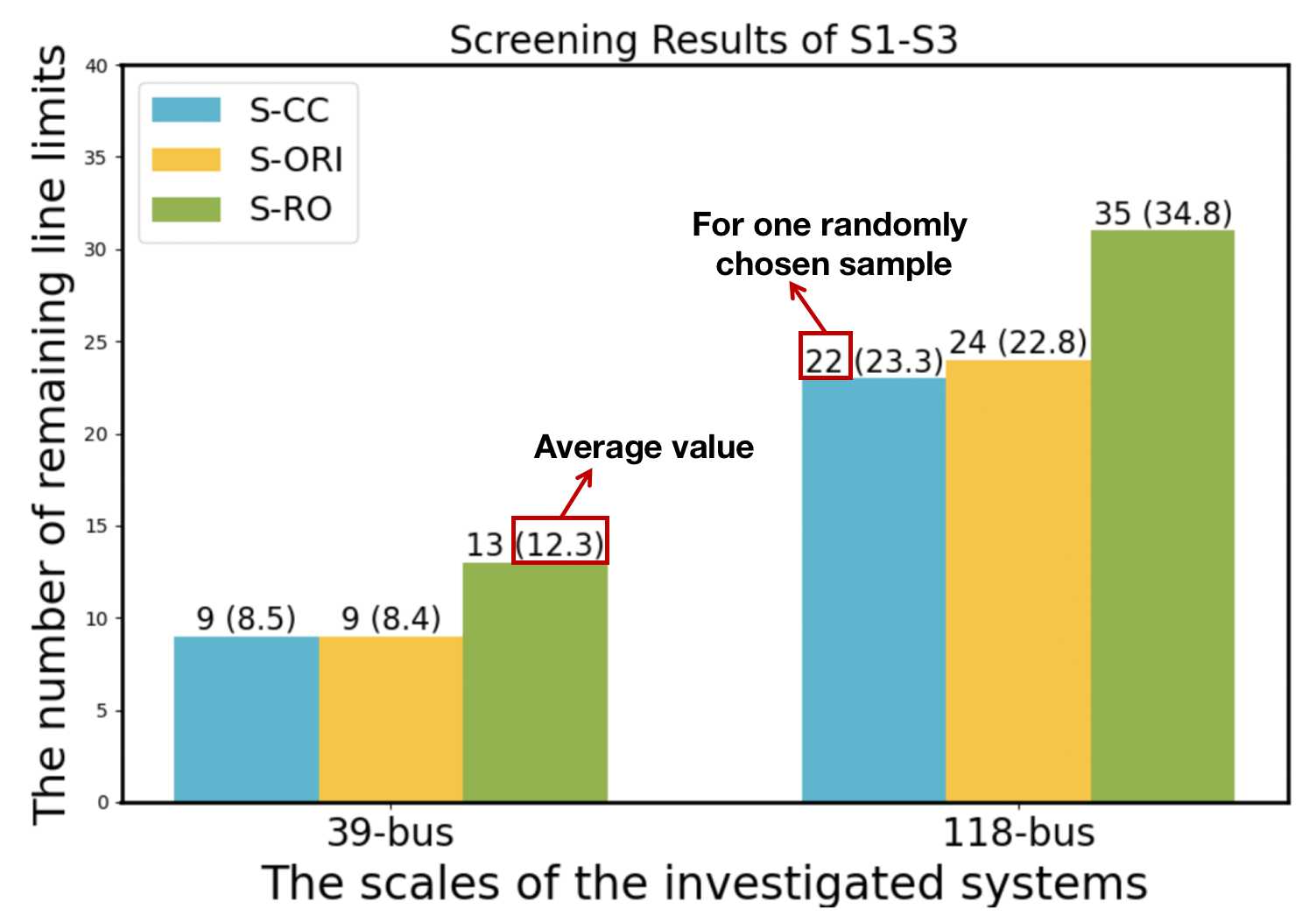}

	\caption{\footnotesize The number of non-redundant line limits given by original, chance-constrained ($\epsilon = 10\%, \sigma_{\omega_i}=1$) and robust screening($\beta_1=0.7, \beta_2 = 1.3$).} \label{num_s1_s3}
\end{figure}
%For both systems the numbers of redundant constraints remain unchanged with smaller $\epsilon$. 

We compare the results of \texttt{S2} with varying settings of $\epsilon$ and $ \sigma_{\omega_i}$ in Fig. \ref{num_s1_s3} and  Table \ref{table2}.
It can be seen that the number of non-redundant limits of the 39-bus system increases to 10 while that of the 118-bus system decreases to 17 with larger $\sigma_{\omega_i}$. This indicates that $\sigma_{\omega_i}$ may have a larger impact than $\epsilon$ on the binding situations of the constraints in the CC-Screening cases. Regarding the size of $S_{RO}$, from Table \ref{table3}, it can be found that using \texttt{S3} will remain more non-redundant constraints for larger variation ranges of uncertainty $\omega$, since there are more binding situations corresponding to the larger load region as shown in Fig. \ref{fig: MPP_region}.

\begin{table}[t]
\centering
\caption{Feasibility Analysis of Testing 2.}\label{table2}
\vspace{-0em}
\setlength{\tabcolsep}{0.35mm}{
% \usepackage{multirow}
\begin{threeparttable}{
\begin{tabular}{c|ccc|ccc}
\hline
 % \textcolor{red}{Can the infea rate use the same format as table iii, e.g., 10.00\%?} 
% \usepackage{multirow}
\multicolumn{1}{c|}{\multirow{2}*{\diagbox{Bus}{Case}}} & \multicolumn{3}{c|}{$\epsilon=5\%,\sigma_{\omega_i}=1$}                                                                     & \multicolumn{3}{c}{$\epsilon=5\%,\sigma_{\omega_i}=10$}                                                                    \\ \cline{2-7} 
                  & \multicolumn{1}{l}{No.Limits} & \multicolumn{1}{l}{Infea. Rate } & \multicolumn{1}{l|}{Solu. Gap } & \multicolumn{1}{l}{No.Limits} & \multicolumn{1}{l}{Infea. Rate} & \multicolumn{1}{l}{Solu. Gap} \\ \hline
39                & 9                             & 0.0\%                              & 0.0\%                              & 10                            & 10.0\%                            & 0.0\%                             \\
118               & 22                            & 0.0\%                             & 0.0\%                              & 17                            & 5.0\%                             & 0.0\%                             \\ \hline
\end{tabular}}
 \begin{tablenotes}
        \footnotesize
        \item * `No. Limits' denotes the number of non-redundant line limits.
        % \item * 'Infe. Rate' means the percentage of infeasible uncertainty realizations.
        % \item * 'Solu. Gap' means the average gap between the original UC cost and the feasible reduced CC-UC cost. 
      \end{tablenotes}
  \end{threeparttable}
}
\vspace{-0em}
\end{table}

\begin{figure}[t]
    % \hspace{0.1cm}
	\centering
	\includegraphics[width=0.99\linewidth]{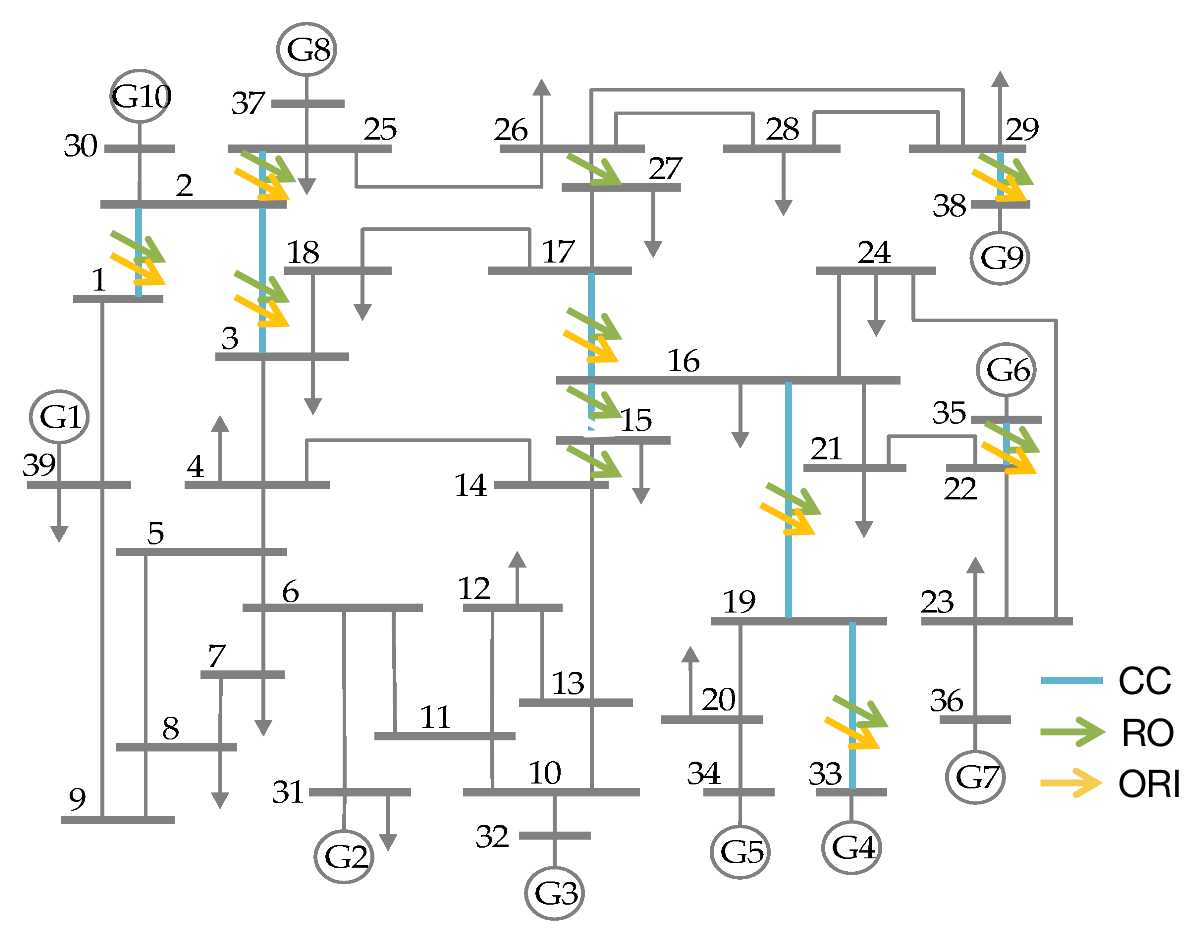}
	\caption{\footnotesize The identified non-redundant line limits of original, chance-constrained ($\epsilon =10\%, \sigma_{\omega_i}=10$) and robust ($\beta_1=0.7, \beta_2 = 1.3$) screening for 39-bus system.} \label{UC_binding_10_10}
\end{figure}

\begin{table*}[ht]
% \vspace{-1em}
\centering
\caption{Feasibility Analysis of Testing 3.} \label{table3}
\vspace{0em}
\setlength{\tabcolsep}{1.4mm}{
\begin{tabular}{c|c|c|c|c}
\hline
\multicolumn{1}{c|}{\multirow{2}*{\diagbox{Bus}{Case}}} & $\beta_1(0.9),\beta_2(1.1)$ & $\beta_1(0.7),\beta_2(1.3)$ & $\beta_1(0.5),\beta_2(1.5)$ & \multirow{2}{*}{Infea. Rate}       \\ \cline{2-4} 
                  & No.Limits                 & No.Limits                 & No.Limits                 & \\ \cline{1-5}
39                & 10                        & 13                        & 13                        &    0.0\%                \\
118               & 30                        & 35                        & 36                        &          0.0\%           \\ \hline
\end{tabular}}
\end{table*}

\begin{table}[t]
\vspace{-0em}
\centering
\caption{Feasibility Analysis of Testing 4.} \label{table4}
\vspace{-0em}
\setlength{\tabcolsep}{0.4mm}{
\begin{threeparttable}{\begin{tabular}{c|ccc|ccc}
\hline
\multicolumn{1}{c|}{\multirow{2}*{\diagbox{Bus}{Case}}}& \multicolumn{3}{c|}{$\epsilon=5\%,\sigma_{\omega_i}=1$}                                                                     & \multicolumn{3}{c}{$\epsilon=5\%,\sigma_{\omega_i}=10$}                                                                    \\ \cline{2-7} 
                  & \multicolumn{1}{l}{No.Limits} & \multicolumn{1}{l}{Infea. Rate} & \multicolumn{1}{l|}{Solu. Gap} & \multicolumn{1}{l}{No.Limits} & \multicolumn{1}{l}{Infea. Rate} & \multicolumn{1}{l}{Solu. Gap} \\ \hline
39                & 9                             & 3.0\%                         & 0.6\%                         & 10                            & 7.5\%                         & 4.4\%                         \\
118               & 22                            & 0.0\%                              & 0.4\%                         & 17                            & 5.0\%                         & 9.5\%                         \\ \hline
\end{tabular}}
  \end{threeparttable}
}
\end{table}

\subsection{Feasibility of Reduced Problems}
After we apply \texttt{S1-S3} for the forcast $\hat{\boldsymbol{\ell}}$ and remove $S_{ORI}$, $S_{CC}$ and $S_{RO}$ from \texttt{M1} respectively offline, we can validate screening performance of model \texttt{T1-T3} by obtaining the UC solutions for the groundtruth $\boldsymbol{\ell}$. For 39-bus system, $S_{ORI}$ are the same as $S_{CC}$ under the setting of $\epsilon=5\%, \sigma_{\omega_i}=1$,  while under the setting of $\epsilon=5\%, \sigma_{\omega_i}=10$, the remaining constraints of \texttt{T2} cover that of \texttt{T1}. As seen from Table \ref{table2}, almost all of the solutions given by \texttt{T2} are feasible, indicating the solution technique is reliable. The infeasible cases mean the removed constraint will be violated in \texttt{M1}, which can be caused by both \texttt{T1} and \texttt{T2}.

% \textcolor{red}{Is this part correct? Shall this be T2 rather than T1?} 
For \texttt{T2}, it appears that when the deviation caused by uncertainty is small, the solution provided by \texttt{T2} can achieve a 0\% infeasibility rate. On the other hand, when $\sigma_{\omega_i}=10$, the infeasibility rate increases to 10\% for the 39-bus system and 5\% for the 118-bus system. For the realizations with feasible solutions, we calculate the gap between the UC costs given by \texttt{T1} and \texttt{T2}. The results for the 118-bus system suggest that, although the indexes of non-redundant constraints in \texttt{T2} are a subset of \texttt{T1}, the solution gap might potentially be 0. This observation implies that the actual active line flows for those feasible uncertainty realizations could be a subset of both \texttt{T1} and \texttt{T2}. 
\begin{figure}[t]
	\centering
	\includegraphics[width=0.93\linewidth]{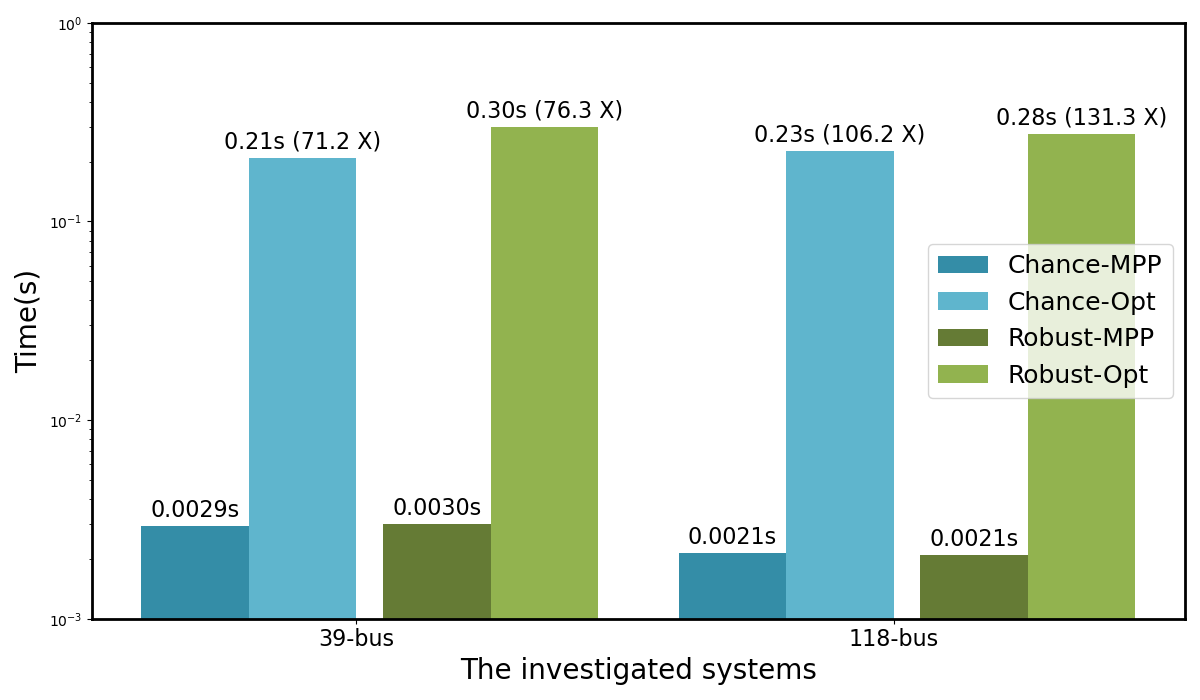}
	\caption{\footnotesize The average screening time for a single line limit using MPP-based and Optimization-based methods. } \label{MPP_time_singe}
\end{figure}

\begin{figure}[h]
	\centering
	\includegraphics[width=0.96\linewidth]{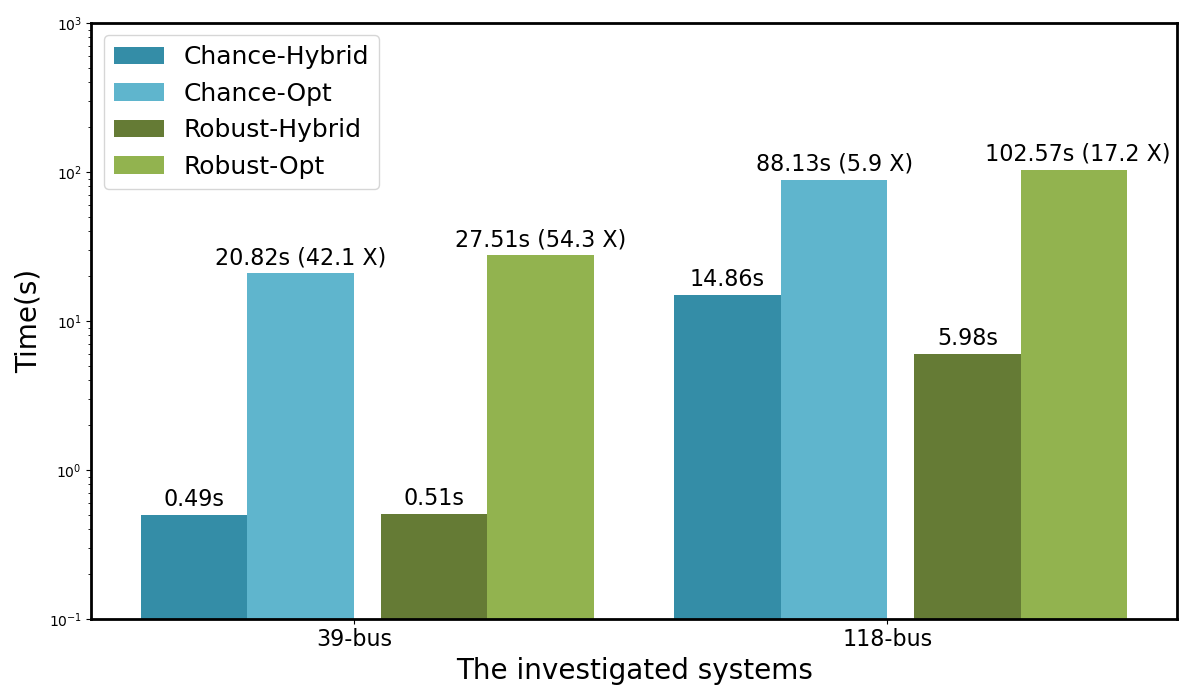}
	\caption{\footnotesize The total screening times of Hybrid and Optimization only methods.} \label{MPP_time_total}

\end{figure}

An alternative way to address the impact of uncertainty on the solution of \texttt{M2} is using the affine generation policy in \eqref{UC:ge_control} for \texttt{T4}, allowing direct generation calculation without solving \texttt{T2} for each uncertainty realization. Results in Table \ref{table4} suggest that infeasibility may arise despite small uncertainty, and for some feasible solutions, like the 118-bus system with $\epsilon=5\%, \sigma_{\omega_i}=1$, a small gap remains between the actual optimal UC cost and the affine policy-determined cost. As the uncertainty increases, the infeasibility rate might exceed the risk level $\epsilon$, and the solution gap could also become larger. Therefore, although the affine generation policy could potentially reduce the online solution time for the ground truth $\boldsymbol{\ell}$, it might also raise the possibility of infeasible solutions.

Regarding the performance of robust screening, Lemma 1 guarantees the feasibility of \texttt{T5}. Besides, the solution provided by \texttt{T3} will be always feasible for \texttt{M1} when $\Tilde{\boldsymbol{\omega}}$ satisfies \eqref{UC:uncertainty_box}. Empirical results in Table \ref{table3} support this analysis, as they demonstrate that, for both 39-bus and 118-bus systems, the solutions given by \texttt{T3} for 200 uncertainty realizations are all feasible for the original UC problem \texttt{M1}, regardless of the three settings for $\beta_1$ and $\beta_2$. Note that though \texttt{T3} offers better feasibility under uncertainty, its solution time is longer than \texttt{T1} and \texttt{T2} due to more remaining constraints in the UC model.

\begin{figure}[t]
    % \hspace{0.15cm}
	\centering
	\includegraphics[width=0.99\linewidth]{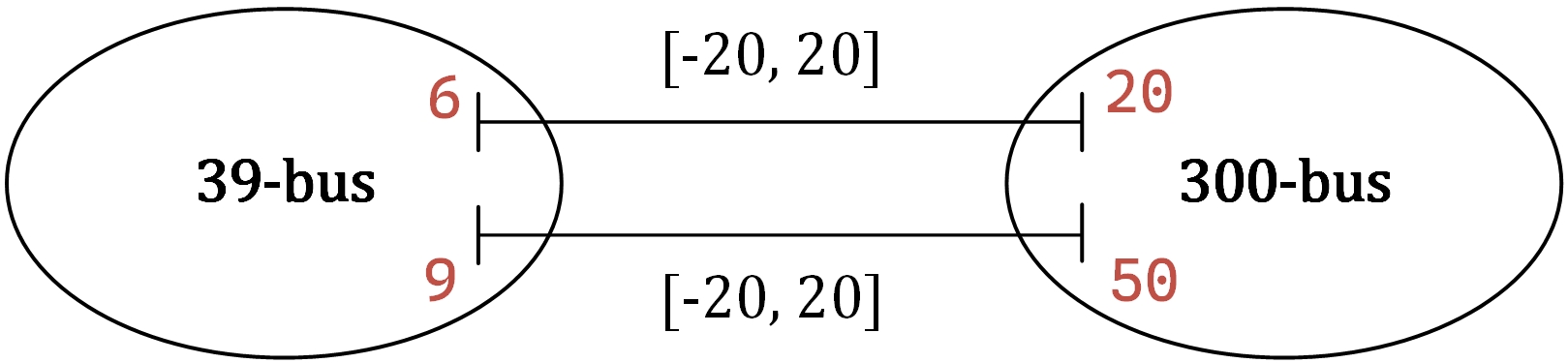}
	\caption{\footnotesize Interconnected power system consists of 39-bus system and 300-bus system.}
 %move the groundtruth $\boldsymbol{\ell}$ from non-binding to binding region and thus causing the screening result of this line limit invalid. }
	\label{fig: S_area}
\end{figure}

\begin{table*}[]
 \centering
\caption{Screening results for the interconnected system.}\label{table_large}
\setlength{\tabcolsep}{6mm}{
\begin{tabular}{c|cc|c}
\hline
\multirow{2}{*}{Area} & \multicolumn{2}{c|}{Number of binding line limits} & \multirow{2}{*}{$\overline{S}_{whole} \subseteq \bigcup_{g=1}^{N_P}\overline{S}_{g}?$} \\ \cline{2-3}
                      & Whole                 & Decomposed                 &                                                                                        \\ \hline
39-bus                & 2045                  & 2068                       & \multirow{2}{*}{$\surd$}                                                               \\
300-bus               & 4933                  & 4933                       &                                                                                        \\ \hline
\end{tabular}
}
\end{table*}

\begin{table*}[]
\centering
\caption{Online screening time of different methods.}\label{table_time}
\setlength{\tabcolsep}{2mm}{
% \usepackage{multirow}
\begin{tabular}{c|cc|ccc}
\hline
\multirow{2}{*}{Area} & \multicolumn{2}{c|}{Total Time (s)}                   & \multicolumn{3}{c}{Time for Single Limit (s)} \\ \cline{2-6} 
                      & Whole Optimization                     & Decomposed Optimization                & Wh-Opt        & De-Opt        & De-MPP        \\ \hline
39-bus                & \multirow{2}{*}{40203.36} & \multirow{2}{*}{37483.16} & 17.61         & 15.42         & 0.06          \\
300-bus               &                           &                           & 16.95         & 16.31         & 0.13          \\ \hline
\end{tabular}
}
\end{table*}

\subsection{Screening Time Comparison }
To investigate the potential of utilizing MPP to accelerate the screening models with different uncertainty formulations, we develop the affine policy for \texttt{S2} and \texttt{S3} via MPP method, respectively. In cases where MPP is not applicable to the screening model or the forecasted load, we revert to solving optimization for \texttt{S2} and \texttt{S3} directly, and we can treat this procedure as \textit{hybrid screening}. Consequently, the total screening time of \textit{hybrid screening} is the sum of using affine policy and initial screening. Fig. \ref{MPP_time_singe} shows that applying the affine policy can accelerate single-line screening by 71.2X to 131.3X, with more significant improvements in the 118-bus system due to its longer initial screening time compared to the 39-bus system. 

Regarding the total screening time, as illustrated in Fig. \ref{MPP_time_total}, the screening procedures of 39-bus and 118-bus systems are accelerated by 5.9X to 54.3X via applying the affine policies. The improvements in total screening time are not as substantial as those for single-line screening, potentially because the initial screening time is much longer than the affine policy. Though hybrid screening can add computation time for a few instances, using MPP still reduces the computation burden a lot.

 To verify the applicability of decomposed screening, we examine the interconnected power system, comprised of a 39-bus system and a 300-bus system (refer to Fig. \ref{fig: S_area}). We gather the screening results of 100 samples using both whole-screening and area-screening methods, as detailed in Table \ref{table_large}. The number of non-redundant constraints identified by the two screening methods is close, with the results of the whole-screening method being a subset of those obtained from the area-screening method. Regarding the screening time, Table. \ref{table_time} shows that the decomposed screening can realize 1.1X acceleration for the total time. For a single line limit, the derivation of MPP policy for the whole system fails, while it can work for the decomposed screening model. The online screening procedure can be accelerated from 125X to 257X via MPP policy for a single line limit of 39-bus or 300-bus system.

\section{Conclusion and Future Work}
In this work, we consider a critical while underexplored constraint screening setting under uncertain electricity demands, along with a novel solution technique for such uncertainty-aware screening problems. The impact of such uncertainty on the screening results, especially the infeasibility of the reduced problems, is investigated both theoretically and empirically under chance-constrained and robust formulations. Further, the multi-parametric program theory and the multi-area screening are demonstrated to be effective for greatly accelerating the screening procedures. This work demonstrates the potential of efficiently and reliably solving UC under forecasting uncertainties. In the future work, it is desirable to assume more general properties of the future load or renewables distribution, which could lead to constraint screening models under distributionally robust settings. We are also interested in investigating the role of data and machine learning for solving UC screening problems.

% \appendix
% \section{Example Appendix Section}


%%       \cite{<label>} ==> [1]

%%
% Example citation, See \cite{lamport94}.


\bibliographystyle{IEEEtran}
\bibliography{bib}

% %% else use the following coding to input the bibitems directly in the
% %% TeX file.

% %% Refer following link for more details about bibliography and citations.
% %% https://en.wikibooks.org/wiki/LaTeX/Bibliography_Management


% %% For numbered reference style

%   \textit{\LaTeX: a document preparation system},
%   Addison Wesley, Massachusetts,

% \end{thebibliography}
\end{document}